\documentclass{amsart}
\usepackage{amsmath, amsfonts}
\usepackage{textcomp}
\usepackage{url}
\usepackage[pdffitwindow = true, pdfstartview = FitH]{hyperref}

\theoremstyle{plain}
\newtheorem{theorem}{Theorem}

\newtheorem{property}[theorem]{Property}
\theoremstyle{definition}
\newtheorem{definition}[theorem]{Definition}
\newtheorem{remark}[theorem]{Remark}

\newcommand{\R}{\mathbb{R}}  %The set of real numbers
\newcommand{\N}{\mathbb{N}}  %The set of natural numbers
\DeclareMathOperator{\conv}{conv} %Convex hull
\newcommand{\leA}{\le_A} 
\newcommand{\leE}{\le_E} 

\newcommand{\BQP}{\textup{BQP}} % Boolean quadratic polytope (Correlation polytope)
\newcommand{\CUT}{\textup{CUT}} % Cut polytope
\newcommand{\SSP}{\textup{SSP}} % Stable set polytope
\newcommand{\VCP}{\textup{VCP}} % Vertex covering polytope
\newcommand{\PACK}{\textup{PACK}} % Set packing polytope
\newcommand{\PART}{\textup{PART}} % Set partition polytope
\newcommand{\DCP}{\textup{DCP}} % Double covering polytope
\newcommand{\AssP}[1]{\textup{#1AP}} % #1-assignment polytope
\newcommand{\DAP}{\AssP{2}} % 2-assignment polytope (Birkhoff polytope)
\newcommand{\TAP}{\AssP{3}} % 3-assignment polytope
\newcommand{\PAP}{\AssP{$p$-}} % n-assignment polytope
\newcommand{\LOP}{\textup{LOP}} % Linear ordering polytope
\newcommand{\QLOP}{\textup{QLOP}} % Quadratic linear ordering polytope
\newcommand{\QAP}{\textup{QAP}} % Quadratic assignment polytope
\newcommand{\QSAP}{\textup{QSAP}} % Quadratic semi-assignment polytope
\newcommand{\TSP}{\textup{TSP}} % Traveling salesman polytope
\newcommand{\ATSP}{\textup{ATSP}} % Traveling salesman polytope

\title{A special role of Boolean quadratic polytopes\\
among other combinatorial polytopes}

\author{Aleksandr Maksimenko}
\address{Laboratory of Discrete and Computational Geometry, P.G. Demidov Yaroslavl State University, ul. Sovetskaya 14, Yaroslavl 150000, Russia} 
\email{maximenko.a.n@gmail.com}

\usepackage{fancyhdr}
\thanks{Supported by the~project No. 477 of P.G. Demidov Yaroslavl State University within State Assignment for Research.}

\begin{document}

\begin{abstract}
We consider several families of combinatorial polytopes
 associated with the~following NP-complete problems:
 maximum cut, 
 Boolean quadratic programming, 
 quadratic linear ordering,
 quadratic assignment,
 set partition, 
 set packing, 
 stable set,
 3-assignment.
For comparing two families of polytopes we use the~following method. 
We say that a~family $P$ is affinely reduced to a~family $Q$
if for every polytope $p\in P$ there exists $q\in Q$ 
such that $p$ is affinely equivalent to $q$ or to a~face of $q$,
where $\dim q = O((\dim p)^k)$ for some constant $k$.
Under this comparison the~above-mentioned families
 are splitted into two equivalence classes.
We show also that these two classes are simpler (in the~above sence)
 than the~families of polytopes of the~following problems:
 set covering,
 traveling salesman,
 0-1 knapsack problem, 
 3-sa\-tis\-fi\-abi\-li\-ty, 
 cubic subgraph,
 partial ordering.
In particular, Boolean quadratic polytopes appear as faces of polytopes 
 in every of the~mentioned families.
\end{abstract}

\maketitle

\section{Introduction}

In 1954, Dantzig, Fulkerson, and Johnson~\cite{Dantzig:1954} solved 
 a~49-city traveling salesman problem via considering 
 a~polytope of this problem.
This idea turned out quite fruitful.
Since then, there were publicated hundreds of papers 
 about properties of various combinatorial polytopes. 
In particular, a~lot of attention was paid to properties of graphs (1-skeletons) of polytopes
 (such as criterion of adjacency, diameter, clique number) 
 and complexities of extended formulations.

In this paper we compare combinatorial characteristics of complexity
 for several families of such polytopes.
The main results was anonced in 2011~\cite{Maksimenko:2011c} and in 2012~\cite{Maksimenko:2012a}, but without the~proofs.

It is natural to consider the~following method for comparing of polytopes.
If a~polytope $p$ is affinely equivalent to a~(not necessary proper) face of a~polytope $q$,
 then $p$ can not be more complicated than $q$ in any reasonable sense.
Below this fact is denoted by $p \leA q$.
If we can compare two polytopes in this sense, 
 then we can compare their characteristics of complexity.
For example, if $p \leA q$ then the~graph of $p$ is a~subgraph of $q$, 
 moreover, the~face lattice of $p$ is embedded into the~face lattice of $q$.
 
We note that in recent times the~most widespread 
 is a~little bit different method for comparing polytopes.
This is related with the~notion of an extended formulation of a~polytope 
 (see for example~\cite{Conforti:2010} and~\cite{Kaibel:2011}). 
A polytope $q$ is called an \emph{extension} (or an \emph{extended formulation}) 
 of a~polytope $p$ if there exists an affine map $\pi$ with $\pi(q) = p$.
The \emph{extension complexity} of a~polytope $p$ is the~size (i.e. number of facets)
 of its smallest extension.
We will denote by $p \leE q$ the~fact that a~(not necessary proper) face of a~polytope $q$
 is an extension of $p$.
It is clear, that 
$$
p \leA q \Rightarrow p \leE q.
$$

For example, it is well known that if $p$ is a convex polytope with $n$ vertices 
and $\Delta_n$ is a simplex with $n$ vertices, then 
\begin{equation}
\label{eq:simplex}
p \leE \Delta_n.
\end{equation}
As a rule, the number of vertices of a combinatorial polytope $p$ is exponential in the dimension $\dim p$. 
Hence, $\Delta_n$ in \eqref{eq:simplex} has exponential dimension
and the comparison \eqref{eq:simplex} becomes useless in practical sense.
So, it is natural to restrict the dimension of $q$ in $p \leE q$ 
by some polinomial of $\dim p$.
More precisely, let $P$ and $Q$ are sets of polytopes,
we will write $P \propto_E Q$
if there exists $k\in \N$ such that for every polytope $p\in P$ there is $q\in Q$ with $p \leE q$ and  $\dim q = O \left((\dim p)^k\right)$.
For example, Yannakakis~\cite{Yannakakis:1991} showed
 that the~matching polytopes and the~vertex packing polytopes 
 are not more complicated (in the~sense of $\propto_E$) than the~traveling salesman polytopes.
In~\cite{Maksimenko:2011a} it was shown that polytopes 
 of any linear combinatorial optimization problem%
\footnote{\label{foot:LCOP}There's considered the following formulation for a linear combinatorial optimization problem. Assume that in a given set $E$, 
each element $e \in E$ has some weight $c(e) \in R$, and
$f:\ 2^E \rightarrow \{0, 1\}$ is a polynomially (with respect to $|E|$) calculable rule. 
Let $S = \{s \subseteq E \mid f(s) = 1\}$ 
be the set of all feasible solutions of the problem. 
We seek for a subset $s \in S$ with the maximal (minimal) summary weight of elements.} 
(among them are cut polytopes, 0-1 knapsack polytopes, 3-sa\-tis\-fi\-abi\-li\-ty polytopes and many other combinatorial polytopes) 
 are not more complicated than the~traveling salesman polytopes.
One year later in~\cite{Maksimenko:2012b} it was shown that polytopes 
 of any linear combinatorial optimization problem 
 are not more complicated than the~cut polytopes and the~0-1 knapsack polytopes.
In general, it seems that this statement is true for the~family of polytopes
 of any known NP-hard problem.
(The results of this paper are another confirmation.) 
That is families of polytopes of NP-hard problems 
 are not distinguishable while comparing by $\leE$.
On the~one hand, this is convenient for obtaining the~results of a~general nature.
For example, since the~extension complexity of cut polytopes is superpolynomial~\cite{Fiorini:2012},
 then the~same is true for (almost) all other families of combinatorial polytopes.
On the~other hand, it is well known that these families are significantly different from each other.
For example, any two vertices of the~cut polytope 
 constitute an edge (1-face) of this polytope 
 (i.e. the~graph of this polytope is complete)~\cite{Padberg:1989}.
Whereas the~checking of nonadjacency of vertices of traveling salesman polytopes is NP-complete~\cite{Papadimitriou:1978}.
So, it is useful to have a~more sensitive method of comparing, like the~mentioned above $\leA$.

With replacing $\leE$ by $\leA$ in the definition of $\propto_E$ we will 
get the definition of \emph{affine reduction} $\propto_A$.
Below we will consider only 0/1-polytopes%
\footnote{0/1-polytope is the convex hull of a subset of the vertices 
of the cube $[0,1]^d$.}.
So, we would like to start with the following example.
It has been shown by Billera and Sarangarajan~\cite{Billera:1996} that
every 0/1-polytope $p \subset \R^d$ with $k$ vertices
is affinely equivalent to a face of the asymmetric traveling salesman polytope:
\[
	p \leA \ATSP_n \quad \text{for } n \ge (4(2^d-k)+1)d. 
\]
Note, that $n$ is exponential in $d$.
Hence, this does note imply the affine reduction of 0/1-polytopes 
to asymmetric traveling salesman polytopes.
Moreover, the family of 0/1-polytopes can not be affinely reduced 
to the family $\{\ATSP_n\}$ for the following reason (see~\cite[p. 12]{Billera:1996}).
There are at least $2^{2^{d-2}}$ combinatorially non-equivalent 
$d$-dimensional 0/1-polytopes for $d\ge 6$~\cite[Proposition 8]{Ziegler:2000}.
On the other hand, if $f$ is the total number of faces of $\ATSP_n$, 
then
\[
	f \le n^2 (n!)^{n^2} \le n^{n^3+2} = 2^{(n^3+2)\log_2 n}.
\]
Therefore, if every $d$-dimensional 0/1-polytope is a face of $\ATSP_n$, then $n$ is superpolynomial in $d$.

In~\cite{Maksimenko:2012c, Maksimenko:2013b} it was shown that
 so-called double covering polytopes are affinely reduced to 
 knapsack polytopes, set covering polytopes, cubic subgraph polytopes, 
 3-SAT polytopes, partial order polytopes, and traveling salesman polytopes.
The linear optimization on double covering polytopes is NP-hard 
 and the~problem of checking nonadjacency on these polytopes is NP-complete~\cite{Matsui:1995}.
Consequently, the~same is true for the~mentioned families.

In this paper we show, that (in the~sense of relation $\propto_A$)
 Boolean quadratic polytopes (and cut polytopes), quadratic linear ordering polytopes, and quadratic assignment polytopes lie in one equivalence class.
Set partition polytopes, set packing polytopes, 
 stable set polytopes, and 3-assignment polytopes
 lie in another (more complicated) equivalence class
 and they are simpler (w.r.t. $\propto_A$) then double covering polytopes.
Thus, the~family of Boolean quadratic polytopes is more pure example 
 of a~family of polytopes associated with NP-hard problems.
They do not have extra details like NP-completness of adjacency relation and the~like.
Naturally, the~following question arises. 
Whether this family is ``the purest'' one?
More precisely, is there some family of polytopes $P$ associated with NP-hard problem such that $P \propto_A \BQP$, but $\BQP \not\propto_A P$?
The answer to this question was provided in~\cite{Maksimenko:2013a}.
One can construct infinitely many families of Boolean $p$-power polytopes
 (Boolean quadratic polytopes is called Boolean $2$-power polytopes)
 such that each $(p+1)$-family is more pure than $p$-family for $p\in\N$, $p \ge 2$.

The rest of the~paper is organized as follows.
Section~\ref{sec:AffBQP} provides a~definition of affine reducibility and its properties.
As an example we show that Boolean quadratic polytopes ($\BQP$) are affinely reduced to stable set  polytopes ($\SSP$), but $\SSP$ can not be affinely reduced to $\BQP$.
In section~\ref{sec:PackPart} it is shown that $\SSP$ are equivalent 
 (in the~sense of affine reduction) 
 to set packing polytopes and set partition polytopes.
In section~\ref{sec:DCP} we consider double covering polytopes ($\DCP$) 
 and prove that $\SSP \propto_A \DCP$, but $\DCP \not\propto_A \SSP$.
In section~\ref{sec:3Ass} it is shown that $\SSP$ are equivalent to 3-assignment polytopes.
In section~\ref{sec:QLOP} we consider quadratic linear ordering polytopes
 and quadratic assignment polytopes
 and show that they are equivalent to $\BQP$.

\section{Affine reducibility
         and Boolean quadratic polytopes}
\label{sec:AffBQP}

Let's consider the~following partial order on the~set of all convex polytopes.

\begin{definition}
\label{DefRel}
The fact that a~polytope $p$ is affinely equivalent to a~polytope $q$
 or to a~face of $q$ will be denoted by $p \leA q$.
If $p$ is affinely equivalent to $q$ itself we will use designation $p =_A q$.
\end{definition}

This relation is useful for estimation of combinatorial characteristics of polytopes.
For example, if $p \leA q$ then the~number of $i$-faces of $p$ 
 is not greater than the~number of $i$-faces of $q$ for $0 \le i \le \dim p$.
It is clear also that the~number of facets of $p$ is not greater than
 the~number of facets of $q$.
Furthermore, the~graph (1-skeleton) of $p$ is a~subgraph of the~graph of $q$.
Hence we can compare clique numbers of these graphs and the~like.
%Here it should be noted that the~clique number is the~lower bound of complexity
% of the~appropriate linear optimization problem 
% in the~class of so-called direct type algorithms \cite{Bondarenko:2008}.
In the~case $p \leA q$ we can also compare 
 the~extension complexities of these polytopes~\cite{Fiorini:2012}.
The same is true for the~rectangle covering bound~\cite{Fiorini:2012}.

On the~other hand, this relation is useless,
 for example, for estimating diameters of graphs of polytopes.
But the~diameter of graph can hardly be seen as a~characteristic of complexity.
It is not greater than 2 for \TSP polytopes~\cite{Padberg:1974}
 and it is equal to 1 for Boolean quadratic polytopes~\cite{Padberg:1989}.
This does not correspond to the~real complexity of these problems.

It turns out that this relation allows to form up 
 the~currently known families of combinatorial polytopes in hierarchical order.
At the~very bottom of this hierarchy there are
 Boolean quadratic polytopes and cut polytopes.

%Название <<булев квадратичный многогранник>> впервые было предложено
% Падбергом в 1989 г. \cite{Padberg:1989}. 
\emph{Boolean quadratic polytope} is the~convex hull of the~set
\begin{equation}
\label{BQP}
	\BQP_n = \left\{x=(x_{ij})\in\{0,1\}^{\frac{n(n+1)}2} \mid 
	        x_{ij} = x_{ii} x_{jj}, \; 1\le i < j\le n\right\}.
\end{equation}
%одновременно являющегося множеством вершин многогранника.
It should be noted that the~notation $\BQP_n$ is commonly used 
for the~convex hull, but for the sake of brevity 
we do not make distinctions between polytopes themselves and the~sets of its vertices.
The same remark applies to all other polytopes discussed below.

%Благодаря простоте формулировки булева квадратичного многогранника 
%и его многочисленным приложениям,
This polytope is also known as a~correlation polytope~\cite{Deza:1997}. 
% and a~clique polytope \cite{Bondarenko:2008}.
Moreover $\BQP_n$ is directly related by so-called covariant mapping 
 with cut polytope, usually denoted by $\CUT_n$~\cite{DeSimone:1990}.
Using the~notation of definition~\ref{DefRel} 
 this relationship may be written as 
 $$\BQP_n =_A \CUT_{n+1}.$$

Let us note that $\BQP_n$ is a face of $\BQP_{n+1}$ 
 defined by $x_{n+1, n+1} = 0$.

\begin{property}
\label{eq:LinBQP}
$\BQP_n \leA \BQP_{n+1}$, $n\in\N$.
\end{property}

Note also that such relations are normal for families of combinatorial polytopes.

In order to illustrate the~basic ideas of this paper, 
we consider another family of polytopes, which is closely related with $\BQP_n$.

Let $G=(V, E)$ is undirected graph with the~set of vertices
$V=\{v_1, v_2, \ldots , v_k\}$ and the~set of edges $E$.
To each vertex $v_i$, $1\le i \le k$, we associate the~component $y_i$ of the~vector
 $y=(y_1, \ldots, y_k)\in\R^k$. 
The \emph{stable set polytope} of a~graph $G$ is the~convex hull of the~set
\begin{equation}
\label{SSP}
	\SSP_k = \left\{y\in\{0,1\}^k \ | \ 
	          y_i + y_j \le 1 \text{ for every edge } \{v_i, v_j\}\in E
\right\}.
\end{equation}
This polytope also known as a~\emph{vertex packing polytope}.
Futhermore, by affine mapping $z_i = 1 - y_i$, $1\le i \le k$,
it is related to the~\emph{vertex covering polytope} of a~graph $G$:
%\begin{equation}
%\label{VCP}
$$
	\VCP_k = \left\{z\in\{0,1\}^k \ | \ 
	          z_i + z_j \ge 1 \text{ for } \{v_i, v_j\}\in E
\right\}.
$$
%\end{equation}
I.e., $\SSP_k =_A \VCP_k$.

Let us note that $\BQP_n$ is uniquely determined for a~fixed $n$, 
 whereas the~notation $\SSP_k$ hides a~set of $k$-dimensional polytopes.
For example, if a~graph $G$ has no edges then $\SSP_k$ is a~cube.
If $G$ is a~complete graph then $\SSP_k$ is a~simplex.
%Разумеется, в контексте настоящей работы комбинаторные свойства таких тривиальных
%объектов как куб или симплекс мало интересны.
Hereinafter the~$\SSP_k$ will be associated with 
 ``the most complicated'' polytope in this set.
More precisely, for a~polytope $p$ and for fixed $k$ inequality
$$
p \leA \SSP_k,
$$
 means that there exists $q\in \SSP_k$ such that $p \leA q$. 

Generalizing this agreement we obtain

\begin{definition}
\label{DefGeneral}
Let $P$ and $Q$ are sets of polytopes.
Then the~record
$$
P \leA Q
$$
indicates that for every $p\in P$ there exists $q\in Q$ such that $p \leA q$.
\end{definition}

This agreement allows us to deduce an analogue of the~property~\ref{eq:LinBQP} for $\SSP_k$.

\begin{property}
%Для каждого многогранника $\SSP_k$, $k\in\N$, среди $\SSP_{k+1}$
%всегда найдется не менее сложный многогранник:
$\SSP_k \leA \SSP_{k+1}$.
\end{property}

Furthermore, using this notation $\BQP_n$ and $\SSP_k$ can be compared.

\begin{theorem}
\label{Theorem1}
$\BQP_n \leA \SSP_k$ for $k = n(n+1)$.
\end{theorem}
(A similar result is given in~\cite{Fiorini:2012}, but they used relation $\leE$ 
 and $k = 2 n^2$.)

\begin{proof}
The equality $x_{ij} = x_{ii} x_{jj}$ in~\eqref{BQP} 
 is equivalent to inequalities
\begin{equation}
\label{eq:Clique}
\begin{aligned}
%  & x_{ij} \in \{0, 1\}, \quad 1\le i \le j \le k, \\
   x_{ii} - x_{ij} &\ge 0, \\
   x_{jj} - x_{ij} &\ge 0, \\
   x_{ii} + x_{jj} - x_{ij} &\le 1,
\end{aligned}
\end{equation}
subject to $x_{ij}\in\{0, 1\}$.
It remains to transform each of them in an inequality of the~form $y_l + y_m \le 1$ 
from~\eqref{SSP}.
For this we introduce $n (n +1)$ new variables:
\begin{equation}
\label{eq:BQP2SSP}
\begin{aligned}
  s_{ij} &= x_{ij},           & 1 &\le i  <  j \le n,\\
  t_{ij} &= x_{ii} - x_{ij},  & 1 &\le i  <  j \le n,\\
  u_i    &= x_{ii},           & 1 &\le i \le n,\\
  \bar{u}_i    &= 1 - x_{ii}, & 1 &\le i \le n.\\
\end{aligned}
\end{equation}
Then the~restrictions~\eqref{eq:Clique} are equivalent to
%\begin{equation}\label{eq:SSP2}
$$
\begin{aligned}
   s_{ij} + \bar{u}_j & \le 1, \\
   t_{ij} + u_j       & \le 1, \\
   u_i    + \bar{u}_i & =   1, \\
   s_{ij} + t_{ij} + \bar{u}_i & = 1,
\end{aligned}
$$
%\end{equation}
subject to integrality of all variables.
Obviously, the~last two equalities (more precisely, $n (n +1) / 2$ equalities)
define some face of a~polytope $\SSP_k$, where $k = n (n +1)$,
defined by the~system of $n (2n - 1)$ inequalities
%\begin{equation}\label{eq:SSP2}
$$
\begin{aligned}
   s_{ij} + \bar{u}_j & \le 1, \\
   t_{ij} + u_j       & \le 1, \\
   u_i    + \bar{u}_i & \le 1, \\
%   s_{ij} + t_{ij}    & \le 1, \\
   s_{ij} + \bar{u}_i & \le 1, \\
   t_{ij} + \bar{u}_i & \le 1.
\end{aligned}
$$
%\end{equation}
Moreover, the~equalities~\eqref{eq:BQP2SSP} connect this face
with the~polytope $\BQP_n$ by nondegenerate affine mapping.
\end{proof}

\begin{remark}
\label{Remark1}
For $k \ge 2$ relation $\SSP_k \leA \BQP_n$ is not satisfied for any $n$.
Since $\BQP_n$ is a~2-neigh\-bor\-ly polytope~\cite{Padberg:1989}.
%Соответственно, все его грани тоже 2-смеж\-ност\-ны.
%С другой стороны, в семействе $\SSP_k$ присутствует
%$k$-мерный куб, который при $k \ge 2$ не является 2-смежностным.
\end{remark}

Relying on definitions~\ref{DefRel} and~\ref{DefGeneral},
we can introduce an analogue of Cook--Karp--Levin polynomial reducibility~\cite{Garey:1979} 
for families of polytopes (as it was done in~\cite{Maksimenko:2012c}).

\begin{definition}
\label{DefAff}
A family of polytopes $P$ 
 is \emph{affinely reduced} to a~family $Q$
 if there exists $r\in \N$ such that
 for every polytope $p\in P$ there is $q\in Q$ with $p \leA q$ 
 and  $\dim q = O \left((\dim p)^r\right)$.
Designation: $P \propto_A Q$.  
\end{definition}

%Отметим, что это определение отличается от определения аффинной сводимости,
% введенного ранее в монографии \cite{Bondarenko:2008}, 
% в сторону усиления условий.

In such terminology the~theorem~\ref{Theorem1} and the~remark~\ref{Remark1} 
 take the~following form:
%Установленные выше факты в такой терминологии приобретают следующий вид:
$$
\BQP \propto_A \SSP, \quad \SSP \lefteqn{\;\not}\propto_A \BQP,
$$
where $\BQP = \{\BQP_n\}$, $\SSP = \{\SSP_k\}$.

We list some obvious properties of this kind of reduction. 
%Перечислим некоторые свойства нового понятия.

\begin{theorem}
\label{ThProp}
Let $P \propto_A Q$. 
Suppose that there are polytopes in $P$ with some of the~following properties:

1) superpolynomial (in the~dimension of a~polytope) number of vertices and facets,

2) superpolynomial clique number of the~graph of a~polytope,

3) NP-completness of nonadjacency relation,

4) superpolynomial extension complexity,

5) superpolynomial rectangle covering bound.

\noindent
Then there are polytopes in $Q$ with the~same properties.
\end{theorem}

%%%%%%%%%%%%%%%%%%%%%%%%%%%%%%%%%%%%%%%%%%%%%%%%%%%%%%%%%%%%%%%%%%%%%%%%%
% 
% Set packing and set partition polytopes
% 
%%%%%%%%%%%%%%%%%%%%%%%%%%%%%%%%%%%%%%%%%%%%%%%%%%%%%%%%%%%%%%%%%%%%%%%%%

\section{Set packing and set partition polytopes}
\label{sec:PackPart}

%Будем придерживаться терминологии монографии \cite{Yemelichev:1981}.

Let $G = \{g_1, \ldots, g_n\}$ be a~finite set
and $S = \{S_1, \ldots, S_d\} \subseteq 2^G$ be a~set of subsets of $G$.
Consider a~subset $T \subseteq S$.
If every $g_i \in G$ belongs no more (no less) than one of elements of $T$
 then $T$ is called a~\emph{packing (covering) of the~set $G$}.
Covering, which is both the~packing, called a~\emph{partition of the~set $G$}.

Let $A = (a_{ij})$ be $n\times d$ matrix of incidences of
 elements of $G$ and elements of $S$:
$$
a_{ij} = 
\begin{cases}
1, &\text{for $g_i\in S_j$,}\\
0, &\text{otherwise.}
\end{cases}
$$
For every subset $T\subseteq S$ we consider its characteristic vector 
 $x=(x_j)\in\R^d$:
$$
x_j = 
\begin{cases}
1, &\text{if $S_j\in T$,}\\
0, &\text{otherwise.}
\end{cases}
$$
Denote the~set of all such characteristic vectors by $\PACK_d = \PACK(S)$.
It is evident that 
$$
\PACK_d = \{x \in \{0,1\}^d \mid A x \le \mathbf{1}\},
$$
where $\mathbf{1}$ is the~$n$-dimensional all 1 vector.
The convex hull of $\PACK_d$ is called the~\emph{set packing polytope}.

Partition polytopes are defined similarly.
The set of vertices $\PART_d$ of the~\emph{set partition polytope} 
satisfies the~equality
\begin{equation}
\label{eq:Part}
A x = \mathbf{1}.
\end{equation}
It is clear that a~set partition polytope is a~face of a~set packing polytope:
\begin{equation}
\label{InePartPack}
  \PART_d \leA \PACK_d.
\end{equation}

Note that a~stable set polytope is a~special case of a~set packing polytope:
$$
  \SSP_k \leA \PACK_d \quad \text{for } d = k.
$$
It is not difficult to prove that
 the~families $\SSP = \{\SSP_k\}$, $\PACK = \{\PACK_d\}$ and $\PART = \{\PART_d\}$
 are equivalent in terms of affine reducibility.
 
\begin{theorem}
\label{ThClass1}
$\SSP \propto_A \PART \propto_A \PACK \propto_A \SSP$.
\end{theorem}
 
\begin{proof}
Show that $\PACK_d$ is a~special case of $\SSP_k$ for $k = d$.
It is sufficient to note that inequality
$$
  x_1 + x_2 + \ldots + x_k \le 1
$$
is equivalent to the~set of inequalities
$$
  x_i + x_j \le 1, \quad 1\le i < j \le k,
$$
 provided $x_i\in \{0,1\}$, $1\le i \le k$.
%Таким образом, семейства $\SSP = \{\SSP_k\}$ и $\PACK = \{\PACK_d\}$ эквивалентны.

The reduction $\PART \propto_A \PACK$ is evident (see~\eqref{InePartPack}).

Now we show that $\SSP \propto_A \PART$.
Consider auxiliary variables
$u_{ij} = 1 - y_i - y_j$, $u_{ij} \in \{0, 1\}$.
Then the~inequalities $y_i + y_j \le 1$ in~\eqref{SSP} 
can be replaced by equalities
$$
  y_i + y_j + u_{ij} = 1.
$$
Consequently,
%\begin{equation}
%\label{IneqSSP_PART}
$$
  \SSP_k \le \PART_d \quad \text{for } d = k + |E| \le k(k+1)/2.
$$  
%\end{equation}
Here $E$ is the~set of edges in~\eqref{SSP}.
\end{proof}

%%%%%%%%%%%%%%%%%%%%%%%%%%%%%%%%%%%%%%%%%%%%%%%%%%%%%%%%%%%%%%%%%%%%%%%%%
% 
% Double covering polytopes
% 
%%%%%%%%%%%%%%%%%%%%%%%%%%%%%%%%%%%%%%%%%%%%%%%%%%%%%%%%%%%%%%%%%%%%%%%%%

\section{Double covering polytopes}
\label{sec:DCP}

The name ``double covering polytopes'' was used in~\cite{Maksimenko:2012c}
 for a~family of polytopes considered in~\cite{Matsui:1995}.

\begin{definition}
\label{DefDouble}
 The \emph{double covering polytope} is the~convex hull of the~set 
  $$
    \DCP_{d} = \{x\in\{0,1\}^d \mid B x = \mathbf{2}\},
  $$
  where $B\in\R^{m\times d}$ is a~0-1 matrix, 
  $\mathbf{2}$ is the~$m$-dimensional all 2 vector,
  and each row of $B$ contains exactly four 1's.
\end{definition}

Previously in~\cite{Maksimenko:2012c, Maksimenko:2013b},
 there have been found the~following relations 
 for several families of combinatorial polytopes 
 with the~property of NP-completeness of nonadjacency relation.

\begin{theorem}[\cite{Maksimenko:2012c, Maksimenko:2013b}]
\label{Theorem2}
The family of double covering polytopes is affinely reduced 
 to families of polytopes associated with the~following problems:
 travelling salesman,
 knapsack, 
 set covering,
 3-satisfiability,
 cubic subgraph,
 partial ordering.
\end{theorem}

Now we prove that stable set polytopes are simpler than double covering polytopes.

\begin{theorem}
$\SSP_k \leA \DCP_d$ for $d = k + |E| + 1$,
where $|E|$ is the~number of edges (inequalities) in the~equation~\eqref{SSP}. 
\end{theorem}

\begin{proof}
Let us look at the~equation~\eqref{SSP}.
For every edge $\{v_i, v_j\}\in E$
 we consider auxiliary variable
 $u_{ij} = 1 - y_i - y_j$, $u_{ij} \in \{0, 1\}$.
%The number of all such variables is not greater than $k(k-1)/2$. 
Thus every inequality $y_i + y_j \le 1$ in~\eqref{SSP} 
 can be replaced by equality
\begin{equation}
\label{PrfSSP-DCP}
  y_i + y_j + u_{ij} = 1.
\end{equation}
Let $u_0$ be yet another auxiliary variable and let $u_0 = 1$.
Hence the~equality~\eqref{PrfSSP-DCP} is equivalent to 
$$
  y_i + y_j + u_{ij} + u_0 = 2.
$$
According to the~definition~\ref{DefDouble}, 
 a~system of such equalities together with the~requirement of integer variables 
 defines the~vertices of a~double covering polytope $\DCP_d$
 for $d = k + |E| + 1$.
The constraint $u_0 = 1$ defines a~face of this polytope.
Moreover this face is affinely equivalent to the~given $\SSP_k$.
\end{proof}

We now show that the~affine reducibility in the~opposite direction is not possible.
Note that the~NP-completness of adjacency relation is inherited when affine reducing
 (theorem~\ref{ThProp}).
The family of double covering polytopes has this property~\cite{Matsui:1995}, 
 whereas for a~stable set polytope the~checking of adjacency
 is polynomial~\cite{Chvatal:1975}.
Hence, if $P \neq NP$ then 
 $\DCP$ can not be affinely reduced to $\SSP$.
It turns out that the~latter is true even without the~assumption $P \neq NP$.

\begin{theorem}
\label{ThDCPnotSSP}
For the~double covering polytope 
$$
P = \conv\{x\in\{0, 1\}^4 \ | \ x_1 + x_2 + x_3 + x_4 = 2\}
$$ 
 the~relation $P \le \SSP_k$ does not hold for any $\SSP_k$.
\end{theorem}

\begin{proof}
The polytope $P$ has 6 vertices 
$$
\begin{aligned}
  x^1 &= (1, 1, 0, 0),\\
  x^2 &= (0, 0, 1, 1),\\
  x^3 &= (1, 0, 1, 0),\\
  x^4 &= (0, 1, 0, 1),\\
  x^5 &= (1, 0, 0, 1),\\
  x^6 &= (0, 1, 1, 0).
\end{aligned}
$$
They are splitted into three pairs with the~following property
\begin{equation}
\label{3xpairs}
x^1 + x^2 = x^3 + x^4 = x^5 + x^6.
\end{equation}

Assume that $P$ is affine equivalent to some face $\conv\{y_1, \ldots, y_6\}$ of $\SSP_k$.
It is clear that the~vertices $y^1, y^2, \ldots, y^6$ inherit the~property~\eqref{3xpairs}:
\begin{equation}
\label{3pairs}
y^1 + y^2 = y^3 + y^4 = y^5 + y^6.
\end{equation}
We now show that there are two more vertices $y^7$ and $y^8$ of $\SSP_k$ with
\begin{equation}
\label{ThGoal}
  y^7 + y^8 = y^1 + y^2.
\end{equation}
This means that the~intersection of $\conv \{y^7, y^8\}$
 and $\conv\{y_1, \ldots, y_6\}$ is not empty.
That is $\conv\{y_1, \ldots, y_6\}$ is not a~face of $\SSP_k$.

For the~vertices $y^1 = (y^1_1, \ldots, y^1_k)$ and $y^2 = (y^2_1, \ldots, y^2_k)$
 we consider the~set
$$
I = \{i\in [k] \ | \ y^1_i = y^2_i\}.
$$
Note that every vertex of $\SSP_k$ is a~0-1 vector.
Thus, by~\eqref{3pairs} and~\eqref{ThGoal},
\begin{equation}
\label{eq:8}
 y^8_i = y^7_i =  \cdots = y^2_i = y^1_i \quad \text{ for } i\in I.
\end{equation}
Therefore we shall consider only those coordinates 
which values are different for every pair of vertices:
$$
 J = \{j\in [k] \ | \ y^1_j + y^2_j = 1\} = [k] \setminus I.
$$

Consider the~six sets
$$
\begin{aligned}
  U &= \{j\in J \ | \ y^1_j = 1\}, &\quad \bar{U} &= \{j\in J \ | \ y^2_j = 1\} = J \setminus U, \\
  V &= \{j\in J \ | \ y^3_j = 1\}, &\quad \bar{V} &= \{j\in J \ | \ y^4_j = 1\} = J \setminus V, \\
  W &= \{j\in J \ | \ y^5_j = 1\}, &\quad \bar{W} &= \{j\in J \ | \ y^6_j = 1\} = J \setminus W. \\
\end{aligned}
$$
All six sets are distinct, 
 otherwise there would be identical vertices among $y^1$, $y^2$, \dots, $y^6$.
%Следовательно, $|J| \ge 3$.
%Воспользуемся тем, 
% что порядок следования пар и порядок вершин внутри пары не существенны.
%Не уменьшая общности, предположим, что множество $U$ является одним из
% наибольших по включению среди данных шести множеств.
%Предположим также, что $U \cap V \cap W \neq \emptyset$.
%Ясно, что мы всегда можем этого добиться с помощью замены
% $V$ на $\bar{V}$ и (или) $W$ на $\bar{W}$.
Under this condition, the~two sets
%При таком условии множества
$$
\begin{aligned}
 S       &= (U \cap V \cap W) \cup (U \cap \bar{V} \cap \bar{W}) \cup (\bar{U} \cap V \cap \bar{W}) \cup (\bar{U} \cap \bar{V} \cap W),\\
 \bar{S} &= (\bar{U} \cap V \cap W) \cup (\bar{U} \cap \bar{V} \cap \bar{W}) \cup (U \cap V \cap \bar{W}) \cup (U \cap \bar{V} \cap W) = J \setminus S
\end{aligned}
$$
 differ from each of the~above six.

Now we can define the~points $y^7$ and $y^8$:
$$
\begin{aligned}
y^7_i &= y^8_i = y^1_i, \quad i\in I,\\
y^7_i &= 1 - y^8_i = 1, \quad i\in S,\\
y^7_i &= 1 - y^8_i = 0, \quad i\in \bar{S}.
\end{aligned}
$$
This 0-1 points differ from $y^1, y^2, \ldots, y^6$
 and equality~\eqref{ThGoal} is satisfied for them.
It remains to prove that $y^7$ and $y^8$ belong to the~$\SSP_k$.
That is if $y_i + y_j \le 1$ holds for $y^1$, $y^2$, $y^3$, $y^4$, $y^5$, $y^6$
 then it holds for $y^7$ and $y^8$ also.

By equation~\eqref{eq:8}, this condition is satisfied for $i,j \in I$.
This is true for $i \in I$ and $j \in J$ also, since 
$$
\max(y^1_i + y^1_j, y^2_i + y^2_j) = y^1_i + 1 = y^7_i + 1 = \max(y^7_i + y^7_j, y^8_i + y^8_j).
$$
It remains to check the~condition for $i,j \in J$.
If $i \in S$ and $j \in \bar{S}$ then $y^7_i + y^7_j = y^8_i + y^8_j = 1$
 and the~condition is fulfilled.

Consider the~case $i,j \in S$.
If $i,j \in U$ then $y^1_i = y^1_j = 1$. 
Hence $y_i + y_j \le 1$ is violated by $y^1$ for $i,j \in U$.
The same is true if $i$ and $j$ both belong to one of the~sets 
  $\bar{U}$, $V$, $\bar{V}$, $W$, $\bar{W}$.
But for any $i$ and $j$ in $S$ the~latter is true.
For example, if $i \in U \cap V \cap W$ and $j \in \bar{U} \cap \bar{V} \cap W$
 then $i,j \in W$ and so on.
Hence for any $i,j\in S$ the~inequality $y_i + y_j \le 1$ is violated
 for at least one of the~vertices $y^1, \ldots, y^6$.

The same is true for the~case $i,j \in \bar{S}$ by symmetry.
\end{proof}

%\begin{theorem}
%Семейство многогранников двойных покрытий не может быть аффинно сведено
% к семейству многогранников независимых множеств:
%$$ 
% \DCP \lefteqn{\;\not}\propto_A \SSP,
%$$
%где $\DCP = \{\DCP_d\}$, $\SSP = \{\SSP_k\}$.
%\end{theorem}

%%%%%%%%%%%%%%%%%%%%%%%%%%%%%%%%%%%%%%%%%%%%%%%%%%%%%%%%%%%%%%%%%%%%%%%%%
% 
% Three index assignment polytopes
% 
%%%%%%%%%%%%%%%%%%%%%%%%%%%%%%%%%%%%%%%%%%%%%%%%%%%%%%%%%%%%%%%%%%%%%%%%%
 
%(3-assignment)
%См. \cite{Bondarenko:2008} на стр. 137.
%http://en.wikipedia.org/wiki/3-dimensional\_matching

\section{Three index assignment polytopes}
\label{sec:3Ass}

Consider a~ground set $S$, $|S|=m$.
Coordinates of a~vector $x\in\R^{S\times S\times S}$
 we denote by $x(s, t, u)$, where $s,t,u\in S$.
Three index assignment (or 3-dimensional matching) problem can be formulated 
 as the~following 0-1 programming problem:
$$
  \sum_{s\in S} \sum_{t\in S} \sum_{u\in S} c(s,t,u) \cdot x(s,t,u) 
  \rightarrow \max,
$$ 
\begin{align}
  &\sum_{s\in S} \sum_{t\in S} x(s,t,u) = 1 \quad \forall u\in S,\label{eq:ThreeBeg}\\
  &\sum_{s\in S} \sum_{u\in S} x(s,t,u) = 1 \quad \forall t\in S,\label{eq:Three2}\\
  &\sum_{t\in S} \sum_{u\in S} x(s,t,u) = 1 \quad \forall s\in S,\label{eq:Three3}\\
  &x(s,t,u) \in \{0, 1\} \quad \forall s, t, u\in S,\label{eq:ThreeEnd}
\end{align}
where $c(s,t,u)\in\R$ is an input vector.
By $\TAP_m$ denote the~set of all vectors $x\in\R^{S\times S\times S}$
 satisfying restrictions~\eqref{eq:ThreeBeg}--\eqref{eq:ThreeEnd}.
The convex hull of $\TAP_m$ is called the~\emph{(axial) three index assignment polytope}.

The first results about this polytope can be found in~\cite{Euler:1987} and~\cite{Balas:1989}.
A more recent survey can be found in~\cite{Qi:2000}.
In the~Russian-language papers there are given the~lower bound for the~clique number
 of the~graph of $\TAP_m$~\cite{Bondarenko:2008}
 and various properties of noninteger vertices
 of relaxations of this polytope (see for example~\cite{Kravtsov:2006}).

It is obvious that $\TAP_m$ is a~special case of $\PART_d$:
\begin{equation}
\label{Ine3AP-PART}
\TAP_m \leA \PART_d \quad \text{for } d=m^3.
\end{equation}
That is the~family of three index assignment polytopes is 
 affinely reduced to set partition polytopes: $\TAP \propto_A \PART$.
 
Using a~standard reduction~\cite{Garey:1979} from 3SAT to 3-dimensional matching, 
 Avis and Tiwary~\cite{Avis:2013} showed that 3SAT polytope is a~projection
 of a~face of a~three index assignment polytope.
That is $3SAT \propto_E \TAP$ in the~sense of relation $\leE$.
However, from inequality~\eqref{Ine3AP-PART}, theorem~\ref{ThClass1}, 
 theorem~\ref{ThDCPnotSSP} and $\DCP \propto_A 3SAT$~\cite{Maksimenko:2012c}
 it follows that the~reduction $3SAT \propto_A \TAP$ is impossible.

Now we show, that $\SSP \propto_A \TAP$.
Therefore $\TAP$ lies in one equivalence class (in the~sense of $\propto_A$) 
 with $\SSP$, $\PART$, and $\PACK$.

For the~graph $G(V, E)$ in the~equation~\eqref{SSP}
we denote by
$$
W = \{v\in V \mid v\notin e \text{ for every }  e\in E\},
$$
the set of isolated vertices.

\begin{theorem}
\label{TheSSP-3AP}
$\SSP_k \le \TAP_m$
 for $m = 3|E| + 2|W|$.
\end{theorem}

\begin{proof}
The ground set $S$ for the~$\TAP_m$
 will consist of three types of elements:
\begin{enumerate}
\item $v$ and $\bar{v}$ for every isolated vertex $v\in W$.

\item $e$ for every edge $e\in E$.

\item $(e, v)$ for every $e\in E$ and $v \in e$.
\end{enumerate}
Now we construct the~set of triples $Q \subset S\times S\times S$ 
 such that the~face
$$
 F = \big\{x \in \conv(\TAP_m) \mid x(q) = 0 \ \forall q \notin Q\big\}
$$ 
of $\conv(\TAP_m)$ is affinely equivalent to the~$\conv(\SSP_k)$.

For every $v\in W$ the~set $Q$ contains four triples:
 $(v, v, v)$, $(\bar{v}, \bar{v}, \bar{v})$,
 $(v, v, \bar{v})$, $(\bar{v}, \bar{v}, v)$.
There are no other triples containing $v$ or $\bar{v}$.
Hence, if $x \in F$ then for every $v\in W$
 we have only two cases: 
$$
 x(v, v, v) = x(\bar{v}, \bar{v}, \bar{v}) = 1
 \quad \text{or} \quad 
 x(v, v, \bar{v}) = x(\bar{v}, \bar{v}, v) = 1.
$$ 

Now we consider elements $e$ and $(e, v)$ of the~set $S$,
 where $e\in E$ and $v \in e$.
For every nonisolated vertex $v \in V \setminus W$ consider 
 the~set of incident edges $E(v) = \{e_{i_1}, \ldots, e_{i_p}\}$,
 where $p=d_G(v)$ is the~degree of $v$.
The set of triples $Q$ contains:
\begin{enumerate} 
\item $\big( e, e, e \big)$ for every $e\in E$.

\item $\big( (e,v), (e,v), (e,v) \big)$ for every $e\in E$ and $v \in e$.

\item $\big( e, e, (e,v) \big)$ for every $e\in E$ and $v \in e$.

\item $\big( (e_{i_q},v), (e_{i_{q+1}},v), e_{i_q} \big)$
   for every nonisolated $v$ and for every $e_{i_q} \in E(v)$,
   where addition $q+1$ meant to be $1 + q \mod p$.
\end{enumerate} 

%Полагая $x(u, v, w) = 0$ для всех остальных троек $(u, v, w)\notin Q$,
% перейдем к рассмотрению соответствующей грани многогранника $\TAP_m$.

We list some properties of the~vertices of the~face $F$.

Note that for every $(e,v) \in T$ the~set $Q$ contains exactly two triples
 with $(e,v)$ in the~third entry: $\big( (e,v), (e,v), (e,v) \big)$ and $\big( e, e, (e,v) \big)$.
Hence, the~equation~\eqref{eq:ThreeBeg} for $u = (e,v)$ is converted into 
$$
x\big( (e,v), (e,v), (e,v) \big) + x\big( e, e, (e,v) \big) = 1.
$$
That is, $x\big( (e,v), (e,v), (e,v) \big)$ is linearly expressed in $x\big( e, e, (e,v) \big)$.

Note also that for every $e \in T$ the~set $Q$ contains exactly three triples
 with $e$ in the~first entry: $\big( e, e, e \big)$, $\big( e, e, (e,v_1) \big)$, and $\big( e, e, (e,v_2) \big)$, where $e = \{v_1, v_2\}$.
Hence, the~equation~\eqref{eq:Three3} for $s = e$ is converted into 
$$
x\big( e, e, e \big) + x\big( e, e, (e,v_1) \big) + x\big( e, e, (e,v_2) \big) = 1.
$$
That is
$x\big( e, e, e \big) = 1 - x\big( e, e, (e,v_1) \big) - x\big( e, e, (e,v_2) \big)$
and 
\begin{equation}
\label{IneProof3AP}
x\big( e, e, (e,v_1) \big) + x\big( e, e, (e,v_2) \big) \le 1.
\end{equation}

Reasoning by analogy, we obtain the~following equation
$$
     x\big( (e_{i_q},v), (e_{i_q},v), (e_{i_q},v) \big) 
   + x\big( (e_{i_q},v), (e_{i_{q+1}},v), e_{i_q} \big) = 1
$$ 
for every nonisolated $v$ and for every $e_{i_q} \in E(v)$,
   where addition $q+1$ is performed modulo $p=d_G(v)$.
Hence,
$$
     x\big( (e_{i_q},v), (e_{i_{q+1}},v), e_{i_q} \big) = 
     1 - x\big( (e_{i_q},v), (e_{i_q},v), (e_{i_q},v) \big) =
     x\big( e_{i_q}, e_{i_q}, (e_{i_q},v) \big).
$$ 
Moreover, since 
$$
     x\big( (e_{i_{q+1}},v), (e_{i_{q+1}},v), (e_{i_{q+1}},v) \big) 
   + x\big( (e_{i_q},v), (e_{i_{q+1}},v), e_{i_q} \big) = 1
$$ 
then
$$
     x\big( e_{i_{q+1}}, e_{i_{q+1}}, (e_{i_{q+1}},v) \big) 
   = x\big( e_{i_q}, e_{i_q}, (e_{i_q},v) \big).
$$ 
That is $x\big( e, e, (e,v) \big) = x\big( e', e', (e',v) \big)$
 for any two edges $e$ and $e'$, $v \in e$, $v \in e'$.

It is not difficult to see that for the~vertices of the~face $F$ all variables $x(s,t,u)$
 are expressed linearly in terms of $x\big( e, e, (e,v) \big)$
 and an inequality~\eqref{IneProof3AP} is an inequality $y_i + y_j \le 1$
 in~\eqref{SSP}.
\end{proof}

\begin{remark}
The obtained results can be easily generalized to the~case
 of $p$ index assignment problem ($p > 3$).
By analogy, the~vertices $\PAP_m$ of an $p$ index assignment polytope
 are 0-1 vectors $x\in\R^{m^p}$.
The coordinates $x_{i_1 i_2 \ldots i_p}$
 ($i_1, i_2, \ldots, i_p \in \{1,2,\ldots,p\}$) satisfy the~following equalities:
$$
%\begin{equation}
\begin{aligned}
  \sum_{i_2, i_3, \ldots, i_p} x_{i_1 i_2 \ldots i_p}      &= 1 \quad \forall i_1\in \{1,\ldots,p\},\\
  \sum_{i_1, i_3, i_4, \ldots, i_p} x_{i_1 i_2 \ldots i_p} &= 1 \quad \forall i_2\in \{1,\ldots,p\},\\
  \ldots\ldots & \ldots\\
  \sum_{i_1, i_2, \ldots, i_{p-1}} x_{i_1 i_2 \ldots i_p}  &= 1 \quad \forall i_p\in \{1,\ldots,p\}.\\
\end{aligned}
$$
%\end{equation}
It is evident that
$$
\PAP_m \le \PART_d, \quad \text{where } d=m^p.
$$
On the~other hand, the~equalities
$$
x_{i_1 i_2 \ldots i_p} = 0 \quad \forall i_p \not= i_{p-1}
$$
 determine a~face of $\PAP_m$ and this face is affinely equivalent to $\AssP{$(p-1)$-}_m$.
Therefore, by theorem~\ref{TheSSP-3AP}
$$
\SSP_k \le \PAP_m \quad \text{for } m = 2k(k-1).
$$
\end{remark}

%%%%%%%%%%%%%%%%%%%%%%%%%%%%%%%%%%%%%%%%%%%%%%%%%%%%%%%%%%%%%%%%%%%%%%%%%
% 
%   Quadratic linear ordering polytopes
% 
%%%%%%%%%%%%%%%%%%%%%%%%%%%%%%%%%%%%%%%%%%%%%%%%%%%%%%%%%%%%%%%%%%%%%%%%%

\section{Quadratic linear ordering polytopes
          and~quadratic assignment polytopes}
\label{sec:QLOP}

%Более точно, задача о квадратичном упорядочении 
% называется <<квадратичной задачей линейного упорядочения>> \cite{Buchheim:2009}.
%То есть, по-существу, является усложнением обычной задачи линейного упорядочения.
We begin by describing the~linear ordering problem
 in terms of graph theory.
 
Let $D = (V, A)$ be a~digraph, 
 where $V = \{1, 2, \ldots, m\}$ is a~vertex set.
We assume that $D$ is complete.
That is $(i,j) \in A$ and $(j,i) \in A$ for any $i,j \in V$, $i \neq j$.
An acyclic tournament%
 \footnote{Each pair of vertices in a~tournament is connected by exactly one arc.}
 in digraph $D$ is called a~\emph{linear ordering}.
%\end{definition} 

Consider a~characteristic vector $y \in \R^{m(m-1)/2}$
 for a~linear ordering $L$ in $D$.
The coordinates $y_{ij}$, $1 \le i < j \le m$, of $y$ are
$$
y_{ij} = 
\begin{cases}
1 &\text{for $(i,j)\in L$,}\\
0 &\text{for $(j,i)\in L$.}
\end{cases}
$$
Denote by $\LOP_m$ the~set of characteristic vectors of all linear orderings in $D$.
The convex hull of $\LOP_m$ is called the~\emph{linear ordering polytope}. 
$\LOP_m$ can also be defined as the~set of integer solutions $y\in\{0,1\}^{m(m-1)/2}$ 
 of the~3-dicycle inequalities (see for example~\cite{Grotschel:1985}):
\begin{equation}
\label{3cycle}
  0 \le y_{ij} + y_{jk} - y_{ik} \le 1, \quad i < j < k.
\end{equation}

In~\cite{Buchheim:2009} the~quadratic linear ordering polytope is defined as follows.
Let
$$
I = \big\{(i,j,k,l) \ | \ i<j, \ k<l, \text{ and } (i,j) \prec (k,l) \big\},
$$
where $(i,j) \prec (k,l)$ means that either $i<k$ or $i=k$ and $j < l$.
For every pair of distinct variables $y_{ij}$ and $y_{kl}$ there is introduced a~new variable
\begin{equation}
\label{eq:QLOP}
  z_{ijkl} = y_{ij} y_{kl}, \quad (i,j,k,l) \in I.
\end{equation}
Denote by $\QLOP_m$ the~set of all vectors $z\in\{0,1\}^d$, 
 $d = \binom{m}{2} \left( \binom{m}{2} + 1 \right) / 2$,
 with coordinates $y_{ij}$ and $z_{ijkl}$ satisfying conditions 
~\eqref{3cycle} and~\eqref{eq:QLOP}.
The convex hull of $\QLOP_m$ is called the~\emph{quadratic linear ordering polytope}.

\begin{theorem}[\cite{Buchheim:2009}] 
$\QLOP_m \leA \BQP_n$ for $n = \binom{m}{2}$.
\end{theorem} 

Buchheim, Wiegele, and Zheng~\cite{Buchheim:2009} 
 exploit this result within a~branch-and-cut algorithm
 for solving the~quadratic linear ordering problem.

We show that an affine reduction in the~opposite direction is also possible.

\begin{theorem}
\label{ThBQP-QLOP}
$\BQP_n \leA \QLOP_m$ for $m = 2n$.
\end{theorem} 

\begin{proof}
The idea of the~proof is simple.
$\LOP_m$ contains an $n$-dimensional cube as a~proper face.
In the~transformation $\LOP_m$ to $\QLOP_m$ this cube turns into a~Boolean quadratic polytope.

%Пусть $m = 2n$. 
Note that equalities $y_{ij} = 0$ and $y_{ij} = 1$ 
 defines supported hyperplanes for $\LOP_m$ and for $\QLOP_m$.
Let
$$
J = \{(2i-1,2i) \mid i\in[n] \}.
$$ 
We set 
\begin{equation}
\label{FaceQube}
y_{ij} = 0 \quad \text{for all } (i,j)\notin J, \ 1 \le i < j \le m.
\end{equation}
Only variables $y_{ij}$ are unfixed where $i$ is odd and $j = i+1$.
Let us check 3-dicycle inequalities~\eqref{3cycle}.
Suppose $i < j < k$, we have two cases:
\begin{enumerate} 
\item If $(i,j)\notin J$ then $y_{ij} = y_{ik} = 0$.
  Thus the~inequality~\eqref{3cycle} 
  is transformed into $0 \le y_{jk} \le 1$. 
\item If $(i,j)\in J$ then $i$ is odd, $j = i+1$ is even, and $k > i+1$.
  Hence, the~inequality~\eqref{3cycle} is equivalent to $0 \le y_{ij} \le 1$.
\end{enumerate} 
Therefore $n$ variables $y_{i\, i+1}$, where $i$ is odd, 
 may take the~values 0 or 1 independently of each other.
Consequently, hyperplanes~\eqref{FaceQube} define a~face of $\LOP_m$ 
 and this face is an $n$-dimensional cube. 

Look at the~variables $z_{ijkl}$, $(i,j,k,l) \in I$.
If $(i,j) \notin J$ or $(k,l) \notin J$ then $z_{ijkl}=0$.
In the~case $(i,j) \notin J$ and $(k,l) \notin J$
 we have $z_{ijkl}=y_{ij} y_{kl}$, 
  and besides $y_{ij}$ and $y_{kl}$ are free variables.

Thus there is the~following nondegenerate affine map between the~face of $\QLOP_m$ and $\BQP_n$:
$$
\begin{aligned}
  x_{ii} &= y_{2i-1,\, 2i}, \quad 1\le i\le n,\\
  x_{ij} &= z_{2i-1,\, 2i,\, 2j-1,\, 2j} = y_{2i-1,\, 2i} \cdot y_{2j-1,\, 2j}, \quad 1\le i < j\le n.
\end{aligned}
$$
%Let $x_{ii}$ in the~equation~\eqref{BQP} be equal to $y_{2i-1, 2i}$, $1\le i\le n$.
%Let $x_{ij} = z_{2i-1, 2i, 2j-1, 2j}=y_{2i-1, 2i} \cdot y_{2j-1, 2j}$, $1\le i < j\le n$.
%Thus the~face of $\QLOP_m$ is affinely equivalent to $\BQP_n$. 
\end{proof}

The story for quadratic assignment polytopes is repeated almost exactly.

The set of vertices $\DAP_m$ of the~\emph{assignment polytope} (or \emph{Birkhoff polytope}) 
 consists of vectors $y\in\{0,1\}^{m\times m}$ satisfying the~conditions
\begin{align}
 \sum_j y_{ij} = 1, \ \forall i \in [m], \label{eq:Ass1}\\ 
 \sum_i y_{ij} = 1, \ \forall j \in [m]. \label{eq:Ass2}
\end{align}

Define new variable $z_{ijkl}$ like~\eqref{eq:QLOP}:
%For every pair $y_{ij}$ and $y_{kl}$ there is introduced a~new variable
\begin{equation}
\label{eq:QAP}
  z_{ijkl} = y_{ij} y_{kl}, \text{ where } (i,j) \prec (k,l).
\end{equation}
Denote by $\QAP_m$ the~set of all vectors $z\in\{0,1\}^d$, $d = m^2 + \binom{m^2}{2}$,
 with coordinates $y_{ij}$ and $z_{ijkl}$ satisfying conditions 
~\eqref{eq:Ass1},~\eqref{eq:Ass2}, and~\eqref{eq:QAP}.
The convex hull of $\QAP_m$ is called the~\emph{quadratic assignment polytope}.
It is also useful to consider the~\emph{quadratic semi-assignment polytope}~\cite{Saito:2009}.
In the~definition of its vertex set $\QSAP_m$ the~condition~\eqref{eq:Ass2} is omitted.

\begin{theorem}[\cite{Rijal:1995, Kaibel:1997, Saito:2009}] 
$\QAP_m \leA \QSAP_m \leA \BQP_n$ for $n = m^2$.
\end{theorem} 

This connection is used in~\cite{Kaibel:1997} for obtaining valid inequalities for $\QAP_m$.
In particular, $\QAP_m$ is a~2-neighborly polytope (every two vertices constitute an edge of it),
since $\BQP_n$ is 2-neighborly.
In~\cite{Kaibel:1997} it is also shown that the~linear ordering polytope $\LOP_m$
 and the~traveling salesman polytope $\TSP_m$ are projectons of $\QAP_m$:
$$
 \LOP_m \leE \QAP_m,  \quad  \TSP_m \leE \QAP_m.
$$ 
Note that the~affine reductions 
 $\LOP \propto_A \QAP$ and $\TSP \propto_A \QAP$ are impossible,
 since $\LOP_m$ is not 2-neighborly for $m \ge 3$~\cite{Young:1978}
 and $\TSP_m$ is not 2-neighborly for $m \ge 6$~\cite{Padberg:1974}.

\begin{theorem}
$\BQP_n \leA \QAP_m$ for $m = 2n$.
\end{theorem} 

\begin{proof}
By analogy with the~proof of theorem~\ref{ThBQP-QLOP}
 it is sufficient to show that the~Birk\-hoff polytope $\DAP_m$
 has an $n$-dimensional cube as a~face.
Let
$$
J = \big\{(i,i) \mid i\in[m]\big\} \cup \big\{(2i-1,2i) \mid i\in[n]\big\} \cup \big\{(2i,2i-1) \mid i\in[n]\big\}.
$$ 
Then the~equalities
$$
 y_{ij} = 0 \text { for every } (i,j) \notin J
$$
define the~required face.
\end{proof}

\section{Resume}

Boolean quadratic polytopes, cut polytopes,
 quadratic linear ordering polytopes, and quadratic assignment polytopes
 are in one class of equivalence within the~framework of affine reducibility.
A bit more complicated class contains stable set polytopes, set packing polytopes, 
 set partitioning polytopes, and $n$-index assignment polytopes for $n \ge 3$.
An even more complicated are double covering polytopes, 3-sa\-tis\-fi\-abi\-li\-ty polytopes,
 set covering polytopes, knapsack polytopes, cubic subgraph polytopes, 
 partial ordering polytopes, traveling salesman polytopes.
The problem of partitioning of these families into equivalence classes is not solved completely.
Nevertheless, Fiorini~\cite{Fiorini:2003} proved that $k$-satisfiability polytopes and
 $m$-satisfiability polytopes are in different classes for $k \not= m$.
Moreover all of them are simpler than traveling salesman polytopes. 
From the~other hand,
 the~families of so-called Boolean $p$-power polytopes also are in different classes
 for distinct values of $p$~\cite{Maksimenko:2013a}.
Besides, Boolean $p$-power polytopes are simpler than Boolean quadratic polytopes 
 (in the~sense of affine reduction).

However if in the~definition of affine reducibility (definition~\ref{DefAff}) 
 we replace the~relation $\leA$ by relation $\leE$
 (recall that we write $p \leE q$ if a~face of a~polytope $q$ is an extension of $p$)
 then all the~mentioned families of polytopes fall into one class of equivalence,
 since the~polytope $P$ of any linear combinatorial optimization problem\footnote{see footnote on the page \pageref{foot:LCOP}} is an affine image
 of a~face of $\BQP_n$, where $n$ is polynomial in the~dimension of $P$~\cite{Maksimenko:2012b}.

Thus the~affine reduction is a~more delicate instrument (versus extending)
 for comparing the~families of combinatorial polytopes.
The most complicated (more precisely, the~richest in its properties) 
 is a~family of traveling salesman polytopes.
Families of Boolean $p$-power polytopes are more simple than any other of the~above.
More precisely they contain the~minimum number of superfluous details 
 (with respect to other families associated with NP-hard problems).
Moreover, apparently, combinatorial and geometric properties 
determining NP-hardness reach the~highest concentration 
precisely in Boolean $p$-power polytopes.

Proceed to a~more precise formulation.
Using the~above results it is easy to derive the~following relations.
\begin{enumerate}
\item $\BQP_n \leA \SSP_k \leA \PACK_k$ for $k = n(n+1)$.
\item $\BQP_n \leA \PART_k$ for $k = 2 n^2$.
\item $\BQP_n \leA \DCP_k$ for $k = 2 n^2 + 1$.
\item $\BQP_n \leA \TAP_k \leA \AssP{$p$-}_k $ for $k = 6 n^2 + 3n$ and $p \ge 3$.
\item $\BQP_n \leA \QLOP_k$ for $k = 2 n$.
\item $\BQP_n \leA \QAP_k$ for $k = 2 n$.
\end{enumerate}
That is any characteristic of complexity of $\BQP_n$ is inherited 
 by the~above families of polytopes.
For example, in 2012 Fiorini et al.~\cite{Fiorini:2012} proved that
 the~extension complexity of $\BQP_n$ is exponential in $n$.
In 2013 the~lower bound was improved to $1.5^n$ by Kaibel and Weltge~\cite{Kaibel:2013}.
Hence, the~extension complexity of $\QLOP_k$ and $\QAP_k$ is also exponential in $k$.
The extension complexity of $\SSP_k$, $\PART_k$, $\DCP_k$, and $\TAP_k$ is $2^{\Omega(n^{1/2})}$.
The same conclusions can be done for the~clique numbers of graphs of the~polytopes,
 since the~clique number for $\BQP_n$ is $2^n$.

%Also we can estimate from below the~clique numbers of graphs of polytopes:

%%%%%%%%%%%%%%%%%%%%%%%%%%%%%%%%%%%%%%%%%%%%%%%%%%%%%%%%
%
%   Bibliography
%
%%%%%%%%%%%%%%%%%%%%%%%%%%%%%%%%%%%%%%%%%%%%%%%%%%%%%%%%

\end{document}